\documentclass[conference]{IEEEtran}

\usepackage{hyperref}
\usepackage{url}
\usepackage{breakurl}
\usepackage{ifpdf}
\usepackage{cite}

\ifpdf
    \usepackage[pdftex]{graphicx}
    \graphicspath{{./fig/}}
    \DeclareGraphicsExtensions{.pdf}
\else
    \usepackage[dvips]{graphicx}
    \graphicspath{{./fig/}}
    \DeclareGraphicsExtensions{.eps}
\fi
\usepackage{bbm}
\usepackage{theorem}
\usepackage[cmex10]{amsmath}
\usepackage[ruled,vlined,titlenumbered]{algorithm2e}
\usepackage{color}

\newcommand{\code}{\ensuremath{\mathcal{C}}}

\newcommand{\dmin}{\ensuremath{d_\mathrm{min}}}
\newcommand{\Hd}[1]{\ensuremath{\mathrm{d}_\mathrm{H}(#1)}}

\newcommand{\F}{\ensuremath{\mathbbm{F}}}
\renewcommand{\vec}[1]{\ensuremath{\mathbf{#1}}}

{\theoremstyle{plain}
   \newtheorem{definition}{Definition}}

{\theoremstyle{plain}
   \newtheorem{theorem}{Theorem}}

\begin{document}

\title{End--to--End Algebraic Network Coding for Wireless TCP/IP Networks}

\IEEEoverridecommandlockouts

\author{
\authorblockN{Christian Senger, Steffen Schober, Tong Mao, Alexander Zeh}\thanks{Christian Senger is supported by DFG,
Germany, under grant BO~867/17.}
\authorblockA{\small Institute of Telecommunications and Applied Information Theory\\
Ulm University, Ulm, Germany \\
\{christian.senger$\;\vert\;$steffen.schober$\;\vert\;$tong.mao$\;\vert\;$alexander.zeh\}@uni-ulm.de}
}

\maketitle

\begin{abstract}
The {\em Transmission Control Protocol (TCP)} was designed to provide reliable transport services in wired networks. In such networks, packet losses mainly occur due to congestion. Hence, TCP was designed to apply congestion avoidance techniques to cope with packet losses. Nowadays, TCP is also utilized in wireless networks where, besides congestion, numerous other reasons for packet losses exist. This results in reduced throughput and increased transmission round--trip time when the state of the wireless channel is bad. We propose a new network layer, that transparently sits below the transport layer and hides non congestion--imposed packet losses from TCP. 
The network coding in this new layer is based on the well--known class of {\em Maximum Distance Separable (MDS)} codes.
\end{abstract}

\IEEEpeerreviewmaketitle

\section{Introduction}\label{sec:intro}

Current versions of TCP \cite{rfc793} like TCP--CUBIC \cite{ha_rhee_xu:2008} and Compound--TCP \cite{tan_song_zhang_sridharan:2005} use elaborate techniques to avoid and to cope with network congestion, which is the main reason for packet losses in wired networks. However, the increasing deployment of wireless networks like IEEE-802.11 (WLAN) or IEEE-802.16 (WiMAX) imposes new challenges on TCP as a reliable transport protocol. In wireless networks, packet losses can occur due to effects like shadowing, interference, and multipath propagation to mention only a few. The problem is, that TCP in all it's current implementations reacts on packet losses with some congestion avoidance strategy, thus reducing throughput and transmission round--trip time.

During the last few years, network coding became a prominent field of interest in the communications- and coding community. The main idea is to view packets as elements of an algebraic structure, e.g. a finite field or a vector space. Then, packet losses correspond to erasures and there exist well--known algebraic techniques to correct them. The most important example of this approach is probably random network coding \cite{ho_koetter_medard_karger_effros:2003} based on the Rank metric \cite{silva_kschischang_koetter:2007, silva_kschischang:2007a, koetter_kschischang:2008}, the corresponding codes are frequently called Rank-- or Gabidulin codes \cite{gabidulin:1985, gabidulin_bossert:2008, bossert_gabidulin:2009}. 
In this paper we follow a different approach for network coding, first proposed by Kabatiansky and Krouk \cite{kabatiansky_krouk_semenov:2005} and further developed and extended by Sidorenko et al. \cite{sidorenko_shen_krouk_bossert:2008}.
It utilizes MDS codes and the traditional Hamming metric. With MDS codes, it is possible to correct errors and erasures, the tradeoff parameter between the two can be selected by the system designer.

The idea of combining TCP with network coding was first presented by Sundararajan et al. \cite{sundararajan_sha_medard_mitzenmacher_barros:2008}, following the Rank metric approach.
We complement their work by following the Hamming metric approach and by introducing a new {\em MDS Network Coding layer (NC layer)} between the IP and TCP layer.
The NC layer is transparent to the established layers which enables an easy deployment.

The paper is organized as follows. In Section~\ref{sec:MDS} we explain MDS codes and show their most important property for network coding. In Section~\ref{sec:layer}
we explain the NC layer together with algorithmic descriptions of the transmitter and receiver sides. The properties and capabilities of TCP with the NC layer are
investigated in Section~\ref{sec:sim} by simulation. We also analyze how network coding--capable nodes affect a mixed network environment.

\section{Maximum Distance Separable (MDS) Codes}\label{sec:MDS}

Let $\vec{a}:=(a_0, \ldots, a_{n-1})$ and $\vec{b}:=(b_0, \ldots, b_{n-1})$ be two vectors over an extended Galois field $\F_q$ with $q=2^m$
where $m$ is a non-negative integer.
The {\em Hamming distance} $\Hd{\vec{a}, \vec{b}}$ is defined as the number of differing coordinates of $\vec{a}$ and $\vec{b}$, i.e.
\begin{equation*}
\Hd{\vec{a}, \vec{b}}:=\vert\{i:a_i\neq b_i, 0\leq i\leq n-1\}\vert.
\end{equation*}

Let $\mathcal C$ be a subspace of $\F_q^n$ with dimension $k$. 
Define the {\em minimum distance $\dmin$} as
\begin{equation*}
\dmin:=\min_{\vec{a}, \vec{b}\in {\mathcal C}}\{\Hd{\vec{a}, \vec{b}}\}.
\end{equation*}
Then, $\code(q;n,k, \dmin) := {\mathcal C}$ is called a $q$-nary linear block code of length $n$, dimension $k$ and minimum distance $\dmin$.

One basic result of Coding Theory is the {\em Singleton Bound}, which bounds the minimum distance for given code length and dimension. The following results were established in \cite{singleton:1964}, see also \cite{macwilliams_sloane:1992}.

\begin{theorem}[Singleton Bound]
The minimum distance of a block code $\code(q; n, k, \dmin)$ is bounded by
\begin{equation*}
\dmin\leq n-k+1.
\end{equation*}
\end{theorem}
\begin{proof}
By definition of $\dmin$ there are no two codewords in $\code$ with $\dmin$ coinciding coordinates, but there might be codewords with $\dmin-1$ coinciding coordinates. Their quantity is at most $q^{n-\dmin+1}$. The total number of codewords must be smaller or equal than this quantity, hence $q^k\leq q^{n-\dmin+1}$.
\end{proof}

An interesting subset of the block codes are those, that have maximum possible minimum distance.

\begin{definition}[MDS Code]
A block code $\code(q; n, k, \dmin)$ is called {\em Maximum Distance Separable (MDS)} if it fulfills the Singleton Bound with equality, i.e.
\begin{equation*}
\dmin= n-k+1.
\end{equation*}
\end{definition}

It can be shown that there are no MDS codes over $\F_2$ besides the repetition codes, the single parity check codes and the codes without any redundancy. Probably, the most important MDS codes are the {\em Reed--Solomon (RS)} codes \cite{reed_solomon:1960, macwilliams_sloane:1992} over $\F_q$ with $q>2$. In our considerations, we will always assume that $q$ is a power of two, i.e. $q=2^m$. In this case, Elements of $\F_q$ can be represented by binary vectors of length $m$. We do not use any of the code's properties in this paper except that they are MDS.

The core of our proposed scheme is the following well--known property of MDS codes.

\begin{theorem}\label{theorem:k-is-enough}
If $\code(q; n, k, \dmin)$ is an MDS code, then any $k$ coordinates of $\vec{c}\in\code$ unambiguously determine $\vec{c}$.
\end{theorem}
\begin{proof}
We have $\dmin=n-k+1$, hence in any codeword at most $k-1$ positions can coincide. Thus any $k$ positions uniquely determine the codeword.
\end{proof}

This allows the following procedure. Assume that a binary message of length $kq$ bit should be transmitted. Let us denote a symbol from $\F_q$ as a {\em segment}. The binary message can be considered as an information vector containing $k$ segments. Theses segments can be encoded into $n\geq k$ segments using an MDS code $\code(q; n, k, \dmin)$. The resulting $n$ segments are transmitted over a network. Then, by Theorem~\ref{theorem:k-is-enough}, the receiver is able to reconstruct the $n$ segments if it receives {\bf arbitrary} $k$ segments.

\section{Network Coding Layer}\label{sec:layer}

We are now ready to explain our proposed scheme for End--to--End algebraic network coding based on MDS codes. It works as a transparent new layer between the Internet and the Transport layer of the Internet protocol stack \cite{rfc1122:1989}, see Figure~\ref{fig:stack}. Transparent means that the functionality of the NC layer is hidden to the existing layers, which simplifies deployment in existing networks.

\begin{figure}[htbp]
\centering
\includegraphics[width=150pt]{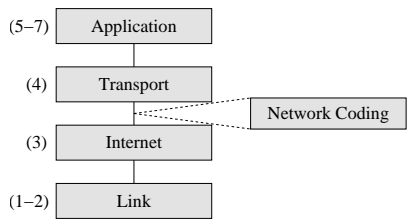}
\caption{The Internet protocol stack extended by the NC layer.}
\label{fig:stack}
\end{figure}

In our considerations, we identify the Transport layer with the TCP layer and the Internet Layer with the IP layer.
For simplicity, we assume a half duplex communication, i.e that payload data is only transmitted in one direction.
Further we assume that the connection setup and termination is always initiated by the transmitter side TCP instance.
The generalization to full duplex is straightforward and does not yield any additional insights.

Connection management segments like SYN or FIN usually pass the NC layer in a transparent manner. However, such segments invoke some special connection management segment handling routines, see Section~\ref{subsec:conman} for details about connection management.

At the transmitter, data segments arriving from the TCP are buffered until $k$ segments are in the buffer. As soon as the 
buffer contains $k$ segments (the information vector), they are encoded\footnote{See Section~\ref{subsec:header} for details 
about which parts of the TCP segments are encoded.} into $n$ segments (the codeword) using an MDS code $\code(q; n, k, \dmin)$.
The parameters $q$ and $k$ of the code must be chosen such that their product $kq$ matches the data segment sizes arriving from the TCP.

To cope with this restriction the NC instance has to change the \emph{Maximum Segment Size (MSS)} parameter \cite{rfc879}
offered by the receiving TCP instance during the connection negotiation If the transmitting TCP instance uses segment
sizes less than the (modified) MSS, the NC layer has to pad the data. 
Another approach would be to implement a strategy for segmenting in the NC layer.
For convenience we assume in the following that the TCP instance always uses segments with size equal to the MSS. 

After encoding, the $n$ segments are passed to the IP layer.
At the receiver, the NC layer acknowledges every received segment from the IP layer. It collects up to $k$ segments in a buffer.
This quantity is sufficient for decoding since $\code$ bears the MDS property. After decoding, the NC layer hands the decoded $k$ segments (the segments which arrived from the TCP at the transmitter) to the receiver side TCP.

Back at the transmitter, the NC layer counts the acknowledgments of coded segments. If their number is greater than some system parameter which we call the {\em speculative ACK threshold}, then the NC layer acknowledges all $k$ data segments to the TCP. This is a {\em speculative acknowledgment} since at this time the receiver side TCP did not necessarily receive all the $k$ data segments -- some might still be stuck within the NC layer. However, the receiver side NC layer will be able to reconstruct the missing data segments from one of the following received codewords. Our simulations show that this happens very fast if the system parameters are chosen appropriately.

The NC layer behavior at the transmitter and receiver is shown in Algorithms~\ref{alg:transmitter} and \ref{alg:receiver}, respectively.
Note that the procedures handle\_TCP\_handshake() and handle\_TCP\_teardown() are both capable of handling incoming SYN segments from TCP and IP.

\begin{algorithm}[htbp]
{%

  \While{true}
  {
   
    $\mathrm{buffer} \leftarrow [\;]$\;
    $\mathrm{ACKcount} \leftarrow 0$\;

    \While{$\vert \mathrm{buffer} \vert < k$}
    {
      receive TCP segment from TCP\;

      \If{TCP segment is SYN}
      {
        handle\_TCP\_handshake()\;
      }
      \ElseIf{TCP segment is FIN}
      {
        handle\_TCP\_teardown()\;
        reset NC layer\;
      }
      \ElseIf{TCP segment is RST}
      {
        forward segment to IP\;
        reset NC layer\;
      }
      \Else
      {
        put TCP segment into buffer\;
      }

    }

    encode $k$ TCP segments into $n$ NC segments\;
    hand over the $n$ NC segments to IP\;

    \While{$\mathrm{ACKcount} < s$}
    {
      receive NC acknowledgment from IP\;
      $\mathrm{ACKcount} \leftarrow \mathrm{ACKcount} + 1$\:
    }

    acknowledge all buffered segments to TCP\;

  }

}
\caption{Network Coding Layer at the Transmitter}
\label{alg:transmitter}
\end{algorithm}

\begin{algorithm}[htbp]
{%

  \While{true}
  {

    $\mathrm{buffer} \leftarrow [\;]$\;

    \While{$\vert \mathrm{buffer} \vert < k$}
    {

      receive NC segment from IP\;

      \If{NC segment is SYN}
      {
        handle\_TCP\_handshake()\;
      }
      \ElseIf{NC segment is FIN}
      {
        handle\_TCP\_teardown()\;
        reset NC layer\;
      }
      \ElseIf{NC segment is RST}
      {
        forward segment to TCP\;
        reset NC layer\;
      }
      \Else
      {
        put NC segment into buffer\;
      }

       hand over NC acknowledgment to IP\;
    }

    decode $k$ NC segments into $n$\;
    extract $k$ information segments, hand over to TCP\;
    discard TCP acknowledgment\;

  }

}
\caption{Network Coding Layer at the Receiver}
\label{alg:receiver}
\end{algorithm}

After the coarse description of the NC layer, we now focus on some details. The NC layer copies most of the TCP layer's functionalities,
 e.g. sequence numbers and the header structure. The following subsections deal with the relations, similarities and differences between the NC and the TCP layer.

\subsection{Header Structure}\label{subsec:header}

The NC layer protocol reuses parts of the TCP header without any modification. To reduce protocol overhead, these common header parts can be 
stripped off the TCP header before encoding of the TCP segments. They can be easily reconstructed at the receiver side by simple extraction from the NC header. The common and stripped--off header parts are source port, destination port, all control flags except ACK, and sequence number. See notes about the latter one in Section~\ref{subsec:seqnum}. The reused and stripped--off header fields are blue shaded in Figure~\ref{fig:headers}.

The remaining TCP header fields are not required by the NC layer. Thus, they are not part of the NC layer header and become part of the encoded TCP segment, i.e. the NC layer payload data. The TCP header fields which are encoded include acknowledgment number, offset, reserved, window, checksum, urgent pointer and options. They are non-shaded on the left-hand side of Figure~\ref{fig:headers}.

Besides the reused TCP header fields, the NC layer adds two additional header fields, i.e. symbol indicator (symb.) and NC options.
The new fields are yellow shaded on the right-hand side of Figure~\ref{fig:headers}. The symbol indicator is responsible
to determine the position of a segment within an MDS codeword, for details see the following section.
The NC options are not used in the current basic version of our protocol. They might be used to signal adaptions
of code rate or speculative ACK threshold in later versions.

\begin{figure}[htbp]
\centering
\includegraphics[width=252pt]{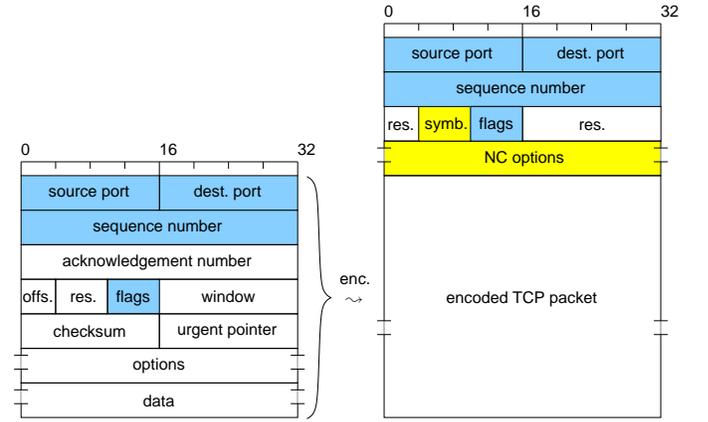}
\caption{Header structure of the TCP (left) and NC (right) layers.}
\label{fig:headers}
\end{figure}

\subsection{Sequence Numbers and Acknowledgments}\label{subsec:seqnum}

We have already mentioned the symbol indicator field in the previous section. In the TCP layer, the sequence number is responsible to order segments at the receiver side. The sequence number is incremented by one for each transmitted segment. Consider an information vector consisting of $k$ TCP segments. The position of each TCP segment is uniquely determined by it's sequence number. But now the information vector is encoded into a codeword consisting of $n$ segments. The NC layer should be able to uniquely determine the position of each segment, this is achieved by the symbol indicator field.

The procedure is as follows: The first $k$ segments of an MDS codeword simply copy the sequence numbers from TCP to their NC sequence number field. Their symbol indicator field is set to all zeros. The remaining $n-k$ segments share the sequence number of the $k$-th segment but their symbol indicator field is incremented by one for each segment. Simple concatenation of sequence number and symbol indicator provides unique identification of each segment's position within an MDS codeword. More precisely, let $\mu$ denote the TCP sequence number and $\nu$ the symbol indicator. Then the value $2^8 \mu+\nu$ uniquely determines the position of each segment.

The length of the symbol indicator field is $8$ bit, hence an MDS codeword can contain at most $256$ redundancy segments. This bounds the minimal code rate by
\begin{equation*}
R=\frac{k}{n}=\frac{k}{k+(n-k)}\geq\frac{k}{k+256}\geq\frac{1}{257},
\end{equation*}
which seems to be sufficient for all practical applications.

No TCP acknowledgments are transmitted over an NC enabled network. On the receiver side, any TCP acknowledgments are discarded immediately. On the transmitter side, the TCP acknowledgments are ''transmitted'' only between the NC and the TCP layers. Hence, the ACK flag is the only flag which is not reused and stripped off before encoding. It is overwritten by the NC layer and used for NC acknowledgment segments, i.e. NC layer segments with empty payload transmitted from the receiver to the transmitter.

\subsection{Connection Management}\label{subsec:conman}
Besides the change of the MSS, the NC layer does not introduce any new routines for connection setup and teardown.
This functions are inherited from the TCP layer. 
If any connection management segment is detected by the NC layer it starts the corresponding routine. 
In case of a TCP teardown or a connection reset, the NC layer schedules a self--reset immediately.
The result of this is that the TCP and the NC layer share a common state for the full duration of a connection.

\section{Simulation Results}\label{sec:sim}

In this section we give simulation results to obtain some insights in the general behavior,
 advantages and disadvantages of the proposed End--to--End algebraic network coding based on MDS codes scheme.
 All simulations were done using the network simulator NS--2 \cite{ns-2} and a basic network topology with one bottleneck link, see Figure~\ref{fig:topology}. Traffic is generated by two FTP sources A1 and A2 which transmit to sinks S1 and S2, respectively. The bottleneck in this setup is link N3$\rightarrow$N4, with a reduced data rate of only 1 Mbit/s. NC layer segments are erased within the network with probability {\em PER (packet erasure rate)}.

\begin{figure}[htbp]
\centering
\includegraphics[width=210pt]{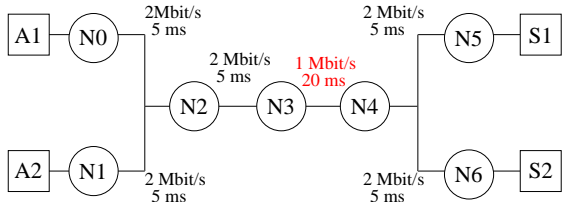}
\caption{Basic network topology for throughput and fairness simulations.}
\label{fig:topology}
\end{figure}

For all simulations we have used an algebraic code with rate $R=k/n=1/2$ and MDS property. We refrain from giving the detailed code parameters here since this section is devoted to the general behavior of the NC scheme.

Figure~\ref{fig:comparison} compares the achievable network throughput in TCP segments per second for a range of PER values. As expected, the NC--enabled TCP is outperformed by traditional TCP in channels with low PER. However, when the channel state deteriorates the traditional TCP throughput quickly falls below the NC--enabled version's throughput, which remains more or less constant over the full PER range. This gives rise to a natural extension of the NC scheme, which adapts the coding rate according to the current state of the channel. The NC options field in the NC layer header can be utilized for this purpose, see Section~\ref{subsec:header}.

\begin{figure}[htbp]
\centering
\includegraphics[width=252pt]{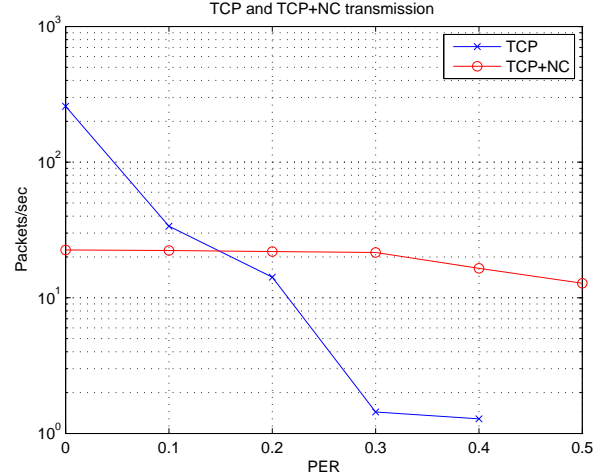}
\caption{Throughput comparison between NC--enabled TCP and traditional TCP for a range of packet erasure probabilities.}
\label{fig:comparison}
\end{figure}

The throughput gain of NC--enabled TCP can be interpreted as follows. In a traditional TCP environment, the protocol copes with increased PER by retransmission of non-acknowledged segments by {\em adaptive repeat request (ARQ)}. This causes a significant amount of delay, especially when the PER is very high. The NC--enabled TCP's strategy to cope with increased PER is to send redundant segments in advance, such that the receiver is not required to wait for requested retransmissions even if segments were erased in the network. Depending on the round trip time and the selected code parameters, this enables the NC--enabled TCP to achieve a higher throughput than traditional TCP as shown by our simulations.

Besides throughput, the most important performance measure of the NC--enabled TCP is fairness, i.e. how fair available network capacities are distributed among the transmitting nodes of a network. We should distinguish between fairness
\begin{itemize}
\item of the NC--enabled TCP against traditional TCP and
\item of a pure NC--enabled environment.
\end{itemize}
The fairness simulation results for both cases are shown in Figure~\ref{fig:nc_vs_tcp} and Figure~\ref{fig:nc_vs_nc}, respectively. In both figures, two connections are active over the bottleneck link. In terms of Figure~\ref{fig:topology} this means that node A1 is transmitting to node S1 and node A2 is transmitting to node S2.

In the NC/TCP case of Figure~\ref{fig:nc_vs_tcp}, it can be seen that NC--enabled TCP does not behave fair against traditional TCP. It dominates the bottleneck link and at any time instance achieves a higher throughput compared to traditional TCP. However, traditional TCP still can use a constant and non-zero share of the bottleneck's capacity, and hence gradual deployment of NC--enabled TCP is still feasible.

\begin{figure}[htbp]
\centering
\includegraphics[width=252pt]{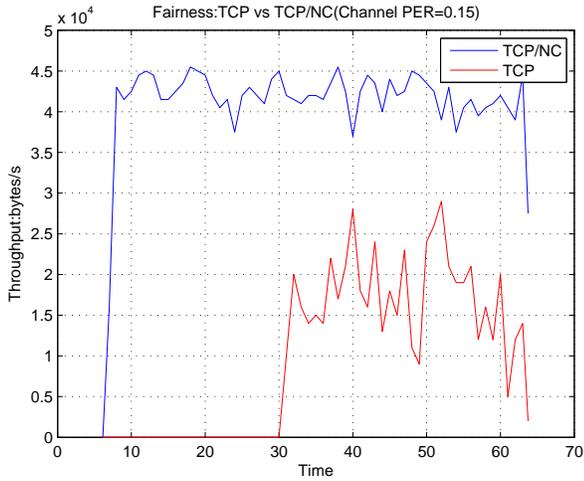}
\caption{Fairness comparison in a mixed environment where the traditional TCP transmission starts before the NC--enabled TCP transmission.}
\label{fig:nc_vs_tcp}
\end{figure}

In the NC/NC case of Figure~\ref{fig:nc_vs_nc}, we observe that after a short balancing period, both NC--enabled TCP links receive approximately the same share of the bottleneck's capacity. Thus, our scheme can be considered as fair in a homogeneous environment.

\begin{figure}[htbp]
\centering
\includegraphics[width=252pt]{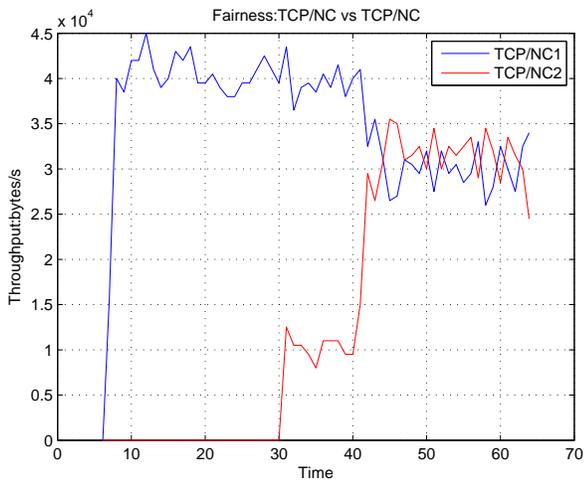}
\caption{Fairness comparison in an NC--enabled TCP environment.}
\label{fig:nc_vs_nc}
\end{figure}

For reference, we also give fairness simulation results for a pure traditional TCP environment in Figure~\ref{fig:tcp_vs_tcp}.

\begin{figure}[htbp]
\centering
\includegraphics[width=252pt]{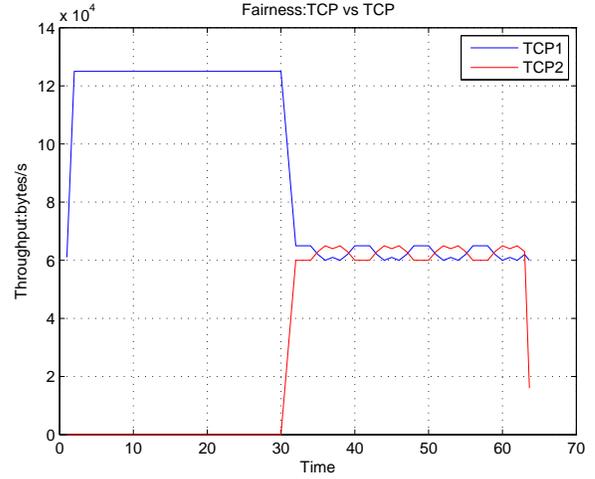}
\caption{Fairness comparison in a traditional TCP environment.}
\label{fig:tcp_vs_tcp}
\end{figure}

\section{Conclusions}
In this paper we proposed a network coding scheme based on MDS codes for application in TCP/IP networks over wireless links. A new layer between the IP and TCP layers was introduced, which hides segment losses caused by the data link layer to the TCP. This is achieved by encoding a set of $k$ TCP segments into a larger set $n$ of encoded NC layer segments and transmitting the larger set. The receiver is then able to reconstruct all $n$ segments by receiving an arbitrary subset of cardinality $k$ by using the MDS property of the used code.

By simulations it was shown that the throughput of a connection is superior compared to traditional TCP if the packet erasure probability is sufficiently large.

Furthermore, the fairness between two TCP streams was studied. It turns out that an NC--enabled TCP stream dominates a traditional TCP stream. However, the fairness property is retained if both streams use NC--enabled TCP.

The proposed network coding scheme rises up an abundance of new research questions. First, optimal code parameters should be derived for realistic settings, i.e. the product of symbol length and dimension should match the average segment size of TCP to avoid padding. Second, the scheme should be modified such that it can adapt to the channel conditions. It is counter-productive to apply encoding of segments when the packet erasure rate is small and in this case the coding rate bounds the achievable throughput. In an adaptive version of the scheme, the code rate might be adapted such that high rates (in the extremal case even code rate $R=1$ which means no encoding at all) are used at low packet erasure rates and smaller code rates are used in the opposite case. Third, one might think about exploiting the error correction capabilities of the utilized MDS codes, i.e. to cope with maliciously introduced packet errors.

\def\noopsort#1{}

\end{document}